\documentclass[showpacs,amsmath,amssymb,onecolumn,pra,superscriptaddress]{revtex4-2}

\usepackage{amsfonts}
\usepackage{amssymb}
\usepackage{amscd}
\usepackage{amsmath}    
\usepackage{multirow}
\usepackage{graphicx}
\usepackage{dcolumn}
\usepackage{bm}
\usepackage{amsthm}
\usepackage{fullpage}
\usepackage{microtype}
\usepackage[libertine]{newtxmath}
\usepackage[tt=false]{libertine} 
\usepackage{caption}
\usepackage{bbm}
\usepackage{hyperref, color}
\hypersetup{colorlinks=true,citecolor=blue, linkcolor=blue, urlcolor=blue}
\usepackage[linesnumbered,boxed,ruled,vlined]{algorithm2e}
\usepackage{bbm}
\usepackage{xcolor}
\usepackage{enumerate}
\usepackage{enumitem}
\usepackage{tabularx}
\usepackage{array}
\usepackage{cleveref}
\newcolumntype{L}[1]{>{\raggedright\arraybackslash}p{#1}}
\newcolumntype{C}[1]{>{\centering\arraybackslash}m{#1}}
\newcolumntype{R}[1]{>{\raggedleft\arraybackslash}p{#1}}

\usepackage{makecell}

\usepackage{cleveref}
\usepackage{color}
\usepackage{graphicx}
\usepackage{float}
\usepackage{bbm}

\usepackage{footnote}

\newtheorem{theorem}{Theorem}[section]

\newtheorem{lemma}[theorem]{Lemma}
\newtheorem{definition}[theorem]{Definition}

\newtheorem{proposition}[theorem]{Proposition}

\renewcommand{\Pr}{P}

\newcommand{\Ext}{\mathrm{Ext}}

\newcommand{\rbs}{\mathrm{bs}}
\newcommand{\rmmneq}{\mathrm{neq}}
\newcommand{\rmmeq}{\mathrm{eq}}

\newcommand{\qstate}{\mathcal{E}}
\newcommand{\gqstate}{\mathcal{E}}

\newcommand{\bra}[1]{\mbox{$\left\langle #1 \right|$}}
\newcommand{\ket}[1]{\mbox{$\left| #1 \right\rangle$}}

\begin{document}
\preprint{APS/123-QED}
\title{Real-Time Seedless Post-Processing for Quantum Random Number Generators}

\begin{abstract}
Quantum-proof randomness extraction is essential for handling quantum side information possessed by a quantum adversary, which is widely applied in various quantum cryptography tasks. In this study, we introduce a real-time two-source quantum randomness extractor against quantum side information.
Our extractor is tailored for forward block sources, a novel category of min-entropy sources introduced in this work. These sources retain the flexibility to accommodate a broad range of quantum random number generators. Our online algorithms demonstrate the extraction of a constant fraction of min-entropy from two infinitely long independent forward block sources.
Moreover, our extractor is inherently block-wise parallelizable, presenting a practical and efficient solution for the timely extraction of high-quality randomness. Applying our extractors to the raw data of one of the most commonly used quantum random number generators, we achieve a simulated extraction speed as high as 64 Gbps.

Keywords: Quantum random number generators; quantum adversary; online two-source extractors; constant extraction rate; infinitely long sources.
\end{abstract}

\author{Qian Li}
\affiliation{Shenzhen lnternational Center For Industrial  And  Applied Mathematics, Shenzhen Research Institute of Big Data, Shenzhen, China.}

\author{Hongyi Zhou}
\email{zhouhongyi@ict.ac.cn}
\affiliation{State Key Lab of Processors, Institute of Computing Technology, Chinese Academy of Sciences, 100190, Beijing, China.}

\maketitle
\section{introduction}
Randomness extraction is a fundamental concept in information theory and cryptography \cite{impagliazzo1989recycle}, playing a key role in various applications, from secure communication to cryptographic key generation. Traditionally, randomness extractors have been designed to efficiently distill uniform randomness from weak-randomness sources against classical side information. However, in the era of quantum technology, the threat posed by quantum adversaries has attracted lots of attention. A quantum adversary has a quantum memory that shares quantum correlations with the randomness source. The information can be extracted by performing various measurements on the quantum memory, which is called quantum side information. An extractor dealing with classical side information may fail in the case of quantum side information \cite{gavinsky2007exponential}. Therefore, the need for quantum-proof randomness extractors that are able to provide information-theoretic security in the presence of a quantum adversary \cite{konig2011sampling} has emerged as a critical concern.

It has been shown that no deterministic extractor can extract even one bit of randomness from one general weak-randomness source with large min entropy, even not from a Santha-Vazirani source \cite{sv1} that is more structural. Therefore, two main approaches to extracting randomness from weak-randomness sources have been explored: seeded extractors and two-source extractors. Seeded extractors rely on a relatively short sequence of random numbers called random seed, which may not be uniform \cite{de2012trevisan} but should be truly random. During the extraction procedure, the seed must be periodically updated to maintain randomness and security, even for strong extractors where the seed can be reused \cite{nisan1996randomness}. Managing and updating the seed securely can be a logistical challenge, especially in long-term deployments or scenarios where the system may be physically tampered with. 
On the other hand, two-source extractors leverage two independent weak-randomness sources to extract high-quality random bits, which avoids the issues above. 

Substantial researches have been dedicated to two-source extractors. The first explicit construction of a two-source extractor appeared in \cite{CG88}, where the authors proposed that randomness can be extracted from two independent sources with an entropy rate of at least $1/2$ based on the Lindsey’s lemma. The following researches mainly aim at relaxing the entropy rate requirement \cite{raz2005extractors,Bou05,barak2006extracting,LiXin13,Li15,Cha16}. In 2016, Chattopadhyay and Zuckerman \cite{Cha16} provided an explicit two-source extractor that works for two independent sources of $n$ bits with an entropy rate of $\log^C n/n$ for a large enough constant $C$, which is nearly optimal. 
The two-source extractor against quantum adversaries was first considered in \cite{multi-source2}, which shows the two-source extractor proposed in \cite{DodisEOR04} is still secure if the quantum side information about the two sources is in the product form. 
This result was further improved in \cite{chung2014multi,FPS16} for more general quantum adversaries.

The existing two-source extractors against quantum adversary either output just one bit or suffer from limitations on the real-time applicability, i.e., the extraction procedure cannot begin until all raw data are generated. The demand for such real-time extractors arises from the fact that the generation rate of random bits is constrained by the post-processing speed \cite{3GQRNG,bai202118}, especially for trusted-device quantum random number generators \cite{ma2016quantum}, which can hinder the timely extraction of high-quality randomness for applications like quantum key distribution.
In response to this challenge, there is a pressing demand for the development of online algorithms capable of extracting randomness in real time. Besides, even worse, we are not aware of any existing quantum-proof two-source extractor that can extract a constant fraction of min-entropy from the two independent sources.

In this research, we propose two \emph{real-time} two-source quantum-proof randomness extractors. In contrast with conventional quantum-proof randomness extractors characterizing the input raw data as min-entropy sources, we find that a large class of quantum random number generators, including one of the most commonly used trusted-device quantum random number generators \cite{gehring2021homodyne}, can be characterized by the so-called \emph{forward block sources} \cite{li2023improved}. This insight allows us to develop more efficient online extraction algorithms. 
Both of our extractors process input raw data on-the-fly, partitioning it into blocks sequentially based on the time order of arrival: they employ the inner product function in specific finite fields on each block. In our first extractor, all blocks share a uniform length that is logarithmic in the overall input raw data length. This property makes the first extractor hardware-friendly and achieve a high extraction speed. While our second extractor utilizes incremental block lengths, rendering it capable of handling infinitely long input raw data. Both extractors can extract a constant fraction of min-entropy from the raw data. Additionally, the independence of applying the inner product function to each block facilitates natural parallel implementation of our extractors.


\section{Preliminaries}
Throughout the paper, we use capital and lowercase letters (e.g., $X$ and $x$) to represent random variables and their assignments, respectively.
We use $U_m$ to represent the perfectly uniform random variable on $m$-bit strings and $\rho_{U_m}$ to represent the $m$-dimensional maximally mixed state. The terms ``quantum adversary” and ``quantum side information" will be used interchangeably.

\begin{definition}[Conditional min-entropy]
Let $Y$ be a classical random variable that takes value $y$ with probability $\Pr_y$ and $\gqstate$ be a quantum system. The state of the composite system can be written as $\rho_{Y\gqstate}=\sum_y \Pr_y\ket{y}\bra{y}\otimes \rho_{\gqstate}^y$, where $\{\ket{y}\}_y$ forms an orthonormal basis. The conditional min-entropy of $Y$ given $\gqstate$ is $H_{\min}(Y|\gqstate)_{\rho_{Y\gqstate}}=-\log_2 p_{\mathrm{guess}}(Y|\gqstate)$, where $p_{\mathrm{guess}}(Y|\gqstate)$ is the maximum average probability of guessing $Y$ given the quantum system $\gqstate$. That is,
\begin{equation}
p_{\mathrm{guess}} (Y|\gqstate)=\max_{\{E_{\gqstate}^y\}_y}\left[\sum_y\Pr_y\mathrm{Tr}\left(E_{\gqstate}^y\rho_{\gqstate}^y\right)\right],
\end{equation}
where the maximization is taken over all positive operator-valued measures (POVMs) $\{E_{\gqstate}^y\}_y$ on $\gqstate$. When $\rho_{Y\gqstate}$ is clear from the context, we will denote the conditional min-entropy as $H_{\min}(Y|\gqstate)$ for brevity.
\end{definition}

In this paper, we call the raw data generated by a QRNG a random source. A general random source is the min-entropy source, where the conditional min-entropy is lower bounded. We consider how to extract randomness from two sequences of raw data generated by two separated QRNGs. Here, we assume that the two separated QRNGs are in product, then the quantum side information is also in the product form. This scenario corresponds to two independent quantum adversaries, each tampering with one of the sources, and then jointly trying to guess the extractor's output. Formally, we have the following definition.
\begin{definition}[Product quantum side information \cite{multi-source2}]
Let $X$ and $Y$ be two independent sources and $\gqstate$ be the quantum side information. We say that the quantum side information is product if 
the joint state $\rho_{XY\gqstate}=\rho_{X\gqstate_1}\otimes\rho_{Y\gqstate_2}$. Note that then $H_{\min}(X| \gqstate)=H_{\min}(X| \gqstate_1)$ and $H_{\min}(Y|\gqstate)=H_{\min}(Y|\gqstate_2)$.
.

\end{definition}

\begin{definition}[Min-entropy quantum-proof two-source extractor]
A function $\Ext:\{0,1\}^{t}\times \{0,1\}^{t}\rightarrow \{0,1\}^m$ is a $(t,k,\epsilon)$ min-entropy quantum-proof two-source extractor, if for every sources $(X,Y)$ and product quantum side information $\gqstate$ where 
$H_{\min}(X| \gqstate)\geq k$ and $H_{\min}(Y| \gqstate)\geq k$, we have
\begin{equation}\label{eq:criteria_ext}
\frac{1}{2}\|\rho_{\Ext(X,Y)\gqstate} - \rho_{U_m}\otimes \rho_{\gqstate} \| \leq \epsilon,
\end{equation}
where $\|\cdot\|$ denotes the trace norm defined by $\|A\| = \mathrm{Tr}\sqrt{A^\dag A}$.
\end{definition}

\section{Main result}

We use $X=X_1X_2\cdots X_N\in (\{0,1\}^b)^N$ and $Y=Y_1Y_2\cdots Y_N\in (\{0,1\}^b)^N$ to denote the raw data generated by two QRNGs, both consisting of $N$ samples each of $b$ bits. For a set $I\subset \mathbb{N}^+$, we write $X_I$ for the restriction of $X$ to the samples determined by $I$. For example, if $I=\{1,3,7\}$, then $X_I=X_1X_3X_7$. We use $\qstate$ to denote the quantum system possessed by the quantum adversary.


\subsection{Model the raw data as forward block sources}
In contrast with conventional quantum-proof randomness extractors characterizing the input raw data as min-entropy sources, we model the input raw data as \emph{forward block sources}, defined as follows.
\begin{definition}[Forward block source]\label{def:rbs}
	A string of random variables $X=X_1\cdots X_N\in(\{0,1\}^b)^N$ is a $(b,N,\delta)$-forward block source given a quantum system $\qstate$ if
for every $1\leq k\leq i\leq N$ and every $x_{1},x_{2},\cdots, x_{k-1}$,
\small{
\begin{equation}\label{eq:rrrbs_def}
\begin{aligned}
H_{\min}(X_k,X_{k+1},\cdots,X_{i}| X_{1}
=x_{1},\cdots,X_{k-1}=x_{k-1},\qstate)
\geq  (i-k+1)\delta\cdot b.
\end{aligned}
\end{equation}
}
\end{definition}
Intuitively, Eq.~\eqref{eq:rrrbs_def} implies each sample $X_i$ brings at least $\delta b$ new randomness. The notion of forward block sources is more special than the min-entropy sources, so it is possible to design extractors with better performance. Fortunately, the notion of forward block sources is still quite general so that it can characterize the input raw data generated by a large class of QRNGs. When the independent and identically distributed (i.i.d.) assumption holds, the raw data straightforwardly satisfies the forward block source. For example, the trusted-devices QRNGs based on single photon detection \cite{Rarity94, Stefanov00, Jennewein00}, vacuum fluctuations \cite{Gabriel10,Symul11,Jofre11}, and photon arrival time \cite{Wahl11,Li13,Nie14} can generate independent raw data satisfying the forward block source. Besides, the raw data of the (semi-)device-independent QRNG under the independent and identically distributed (i.i.d.) assumption also satisfies the forward block source if non-zero min-entropy lower bound is calculated in test rounds.
Moreover, for QRNGs with correlated raw data, one can construct appropriate physical models to check whether Eq.~\eqref{eq:rrrbs_def} is satisfied. For example, 
the ones based on homodyne detection \cite{gehring2021homodyne} with finite bandwidth can both be shown to be a forward block source \cite{li2023improved}. Finally, we remark that there does not exist any nontrivial deterministic extractor for a single block source.



\subsection{Gadget: a construction of min-entropy quantum-proof two-source extractor}

In this subsection, we provide a construction of min-entropy quantum-proof two-source extractor (see Definition \ref{def:IP}), which will be used as a gadget. 
Let $\mathbb{F}_{2^q}$ denote the finite field on $2^q$ elements. We identify $\mathbb{F}_{2^q}$ with $\{0,1\}^q$ and $\mathbb{F}_{2^q}^n$ with $(\{0,1\}^{q})^n$ where the addition operation $``+_q"$ in the field $\mathbb{F}_{2^q}$ is the bit-wise parity operation in $\{0,1\}^{q}$.  


\begin{definition}[Inner-product extractor]\label{def:IP}
The $(n,q)$ inner-product extractor $\Ext_{IP}^{n,q}:(\{0,1\}^{q})^n\times (\{0,1\}^{q})^n\rightarrow \{0,1\}^q$ is defined as
$\Ext_{IP}^{n,q}(x,y)=\sum_{i=1}^n x_i y_i$. Here, $x=x_1\cdots x_n\in(\{0,1\}^q)^n$ and $y=y_1\cdots y_n\in(\{0,1\}^q)^n$, and the addition operations and multiplication operations are in the field $\mathbb{F}_{2^q}$. 

\end{definition}
\begin{theorem}\label{thm:garget}
$\Ext^{n,q}_{IP}$ is a $(qn$,$\delta qn$,$\sqrt{3}\cdot 2^{-\frac{1}{4}-\left(\frac{\delta}{4}-\frac{1}{8}\right)qn+2q})$ min-entropy quantum-proof two-source extractor.
\end{theorem} 

Theorem \ref{thm:garget} immediately follows from the following two lemmas, namely Lemma \ref{lem:xor} and Lemma \ref{lem:one-bit}.
\begin{lemma}[Classical-Quantum XOR-Lemma \cite{multi-source2}]\label{lem:xor}
Let $Z$ be a classical random variable of $q$ bits, and $\gqstate$ be a quantum system. Then 
\[
\|\rho_{Z\gqstate}-\rho_{U_q}\otimes\rho_{\gqstate}\|\leq 2^{q}\sum_{S\in\{0,1\}^q\setminus\{0^q\}} \|\rho_{(S\cdot Z)\circ\gqstate}-\rho_{U_1}\otimes\rho_{\gqstate}\|.
\]
Here $S\cdot Z:=\sum_{i=1}^n S_i\cdot Z_i$ (mod $2$) is the inner product over $\mathbb{F}_2$. 
\end{lemma}
The following lemma will also be used, whose proof can be found in  Appendix \ref{sec:one-bit}.
\begin{lemma}\label{lem:one-bit}
Let $X,Y$ be random variables on $(\{0,1\}^q)^n$ and $\gqstate=\gqstate_1\otimes\gqstate_2$ be the product quantum side information where $H_{\min}(X|\gqstate_1)\geq \delta qn$ and  $H_{\min}(Y|\gqstate_2)\geq \delta qn$.
Let $Z=\Ext_{IP}^{n,q}(X,Y)$. Then for any $S\in\{0,1\}^q\setminus\{0^q\}$, we have 
\[
\frac{1}{2}\|\rho_{(S\cdot Z)\circ\gqstate}-\rho_{U_1}\otimes\rho_{\gqstate}\|\leq \sqrt{3}\cdot 2^{-\frac{1}{4}-\left(\frac{\delta}{4}-\frac{1}{8}\right)qn}.
\]
\end{lemma}

\subsection{Online Seedless Extractors for forward block sources}\label{sec:ext_rbs}

In this subsection, we design two online seedless quantum-proof extractors that can both extract a constant fraction of the min-entropy from forward block sources. The two extractors both proceed in the following fashion: partition the input raw data $X,Y$ into blocks, and apply the inner-product extractor $\Ext_{IP}$ to each block separately. Thus the two extractors can be paralleled naturally. The two extractors are described in Algorithms \ref{alg:rbseq} and \ref{alg:rbsneq} respectively. Given the input raw data $X,Y$ each of length $N$, the first extractor, named $\Ext_{\rbs}^{\rmmeq}$, evenly partitions the input raw data into blocks each of size $O(\log N)$, and the second extractor, named $\Ext_{\rbs}^{\rmmneq}$, requires incremental block lengths. Compared to the first extractor, the second one is less hardware-friendly and sacrifices the extraction speed in general. On the other hand, it enjoys the property that it can deal with infinite raw data, without the need to determine the raw data length prior to extraction.

\begin{algorithm}[h]
	\caption{$\Ext_{\rbs}^{\rmmeq}$}
	\label{alg:rbseq}
	\textbf{Input}: Two independent $(b,N,\delta)$-forward block sources $X=X_1\cdots X_N$ and $Y=Y_1\cdots Y_N$. And $0<\epsilon<1$;\\
     Let $n:=\left\lceil\frac{24}{2\delta-1}\right\rceil$, $q:=b\left\lceil \log_2 \left(\frac{N}{\epsilon n}\right)/b\right\rceil$, and $i:=1$;\\
	\For{$\ell=1$ to $Nb/qn$}{
		Let $I_{\ell}:=[i,i-1+qn/b]$;\\
		Compute $Z^{(\ell)}:=\Ext_{IP}^{n,q}(X_{I_{\ell}},Y_{I_{\ell}})$;\\
		Let $i:=i+n$;\\
		Output $Z^{(\ell)}$.
	}
\end{algorithm}

\begin{algorithm}[htb]
	\caption{$\Ext_{\rbs}^{\rmmneq}$}
	\label{alg:rbsneq}
	\textbf{Input}: Two $(b,\infty,\delta)$-forward block sources $X=X_1\cdots X_\infty$ and $Y=Y_1\cdots Y_\infty$ independent conditioned on quantum adversary $\gqstate$;\\
	\textbf{Parameter:} $q_1,\Delta\in\mathbb{N}^+$; \\
	Let $i:=1$ and $n:=\left\lceil\frac{24}{2\delta-1}\right\rceil$;\\
	\For{$\ell=1$ to $\infty$}{
		Let $I_{\ell}:=[i,i-1+q_{\ell}n/b]$;\\
		Compute $Z^{(\ell)}:=\Ext_{IP}^{n,q_{\ell}}(X_{I_{\ell}},Y_{I_{\ell}})$;\\
		Let $i:=i+q_{\ell}n/b$ and $q_{\ell+1}:=q_{\ell}+\Delta\cdot b$;\\
		Output $Z^{(\ell)}$.
	}
\end{algorithm}

Note that if we impose $\Delta=0$ and let $q_1=b\left\lceil \log_2 \left(\frac{N}{\epsilon n}\right)/b\right\rceil$, then $\Ext_{\rbs}^{\rmmneq}$ becomes $\Ext_{\rbs}^{\rmmeq}$.
We first analyze the second extractor $\Ext_{\rbs}^{\rmmneq}$.

\begin{figure}[h]
	\centering
	\includegraphics[scale=0.6]{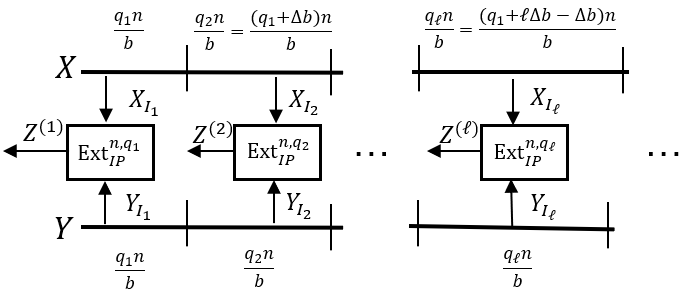}
	\caption{Illustration of $\Ext_{\rbs}^{\rmmneq}$.}
	\label{fig1}
\end{figure}

\begin{theorem}\label{thm:ext_rbs}
The extractor $\Ext_{\rbs}^{\rmmneq}$ satisfies the following properties. For any $k\in\mathbb{N}$,
\begin{equation}
\begin{aligned}		
\frac{1}{2}\|\rho_{Z^{(1)}\circ Z^{(2)}\circ\cdots\circ Z^{(k)}\qstate}-\rho_{U_{m_k}}\otimes \rho_{\qstate}\|
\leq  \sqrt{3}\cdot \sum_{\ell=1}^k2^{-\frac{1}{4}-\left(\frac{\delta}{4}-\frac{1}{8}\right) q_\ell n+2q_\ell}\leq \sqrt{3}\cdot \sum_{\ell=1}^k2^{-\frac{1}{4}-q_\ell}
<\sqrt{3}\cdot \frac{2^{-\frac{1}{4}-q_1}}{1-2^{-\Delta b}},
\end{aligned}
\end{equation}
where $m_k=\sum_{\ell=1}^{k} q_\ell =kq_1+\frac{(k-1)k\Delta b}{2}$.
\end{theorem}

\begin{proof} Fix $k\in\mathbb{N}$. For convenience of presentation, for any $0\leq \ell\leq k$, we use $Z^{>\ell}$ to represent $Z^{(\ell+1)}\circ Z^{(\ell+2)}\circ\cdots\circ Z^{(k)}$, $X^{\leq \ell}$ to represent $X_{I_1}X_{I_2}\cdots X_{I_{\ell}}$, and $Y^{\leq \ell}$ to represent $Y_{I_1}Y_{I_2}\cdots Y_{I_{\ell}}$. Besides, let $\tilde{q}_{\ell}:=q_{\ell+1}+q_{\ell+2}+\cdots+q_{k}$.

 In fact, we will prove that for any $0\leq \ell\leq k$, it has
\begin{equation}\label{eq:induction} 
\begin{aligned}
\frac{1}{2}\left\|\rho_{Z^{>\ell}X^{\leq \ell}Y^{\leq \ell}\qstate}-\rho_{U_{\tilde{q}_\ell}}\otimes \rho_{X^{\leq \ell}Y^{\leq \ell}\qstate}\right\|\leq \sqrt{3}\cdot \sum_{j=\ell+1}^k 2^{-\frac{1}{4}-\left(\frac{\delta}{4}-\frac{1}{8}\right)q_j n+2q_j},
\end{aligned}
\end{equation}
which implies the theorem immediately by letting $\ell=0$.

The proof is by an induction on $\ell$. The base case when $\ell=k$ is trivial. The induction proceeds as follows. Assume Eq.~\eqref{eq:induction} is true. Then due to the contractivity of trace-preserving quantum operations, we have
\begin{equation}\label{eq:proof1}
\begin{aligned}
	\frac{1}{2}\big\|\rho_{Z^{>\ell}\Ext_{IP}\left(X_{I_{\ell}},Y_{I_{\ell}}\right)X^{\leq \ell-1}Y^{\leq \ell-1}\qstate}-\rho_{U_{\tilde{q}_\ell}}\otimes \rho_{\Ext_{IP}\left(X_{I_{\ell}},Y_{I_{\ell}}\right)X^{\leq \ell-1}Y^{\leq \ell-1}\qstate}\big\|
\leq \sqrt{3}\cdot \sum_{j=\ell+1}^k 2^{-\frac{1}{4}-\left(\frac{\delta}{4}-\frac{1}{8}\right)q_j n+2q_j}.
	\end{aligned}
\end{equation}

	On the other hand, by Definition \ref{def:rbs}, for any assignment $x^{\leq \ell-1}$ of $X^{\leq \ell-1}$ and $y^{\leq \ell-1}$ of $Y^{\leq \ell-1}$, we have
	\begin{equation}
 \begin{aligned}
	H_{\min}(X_{I_{\ell}}| X^{\leq \ell-1}=x^{\leq \ell-1},\qstate)\geq \delta q_{\ell}n,  \mbox{ and }
 H_{\min}(Y_{I_{\ell}}| Y^{\leq \ell-1}=y^{\leq \ell-1},\qstate)\geq \delta q_{\ell}n.
 \end{aligned}
	\end{equation}
	Then, recalling that $\Ext_{IP}$ is a min-entropy quantum-proof two-source extractor (Theorem \ref{thm:garget}), it follows that
\begin{align*}	
\frac{1}{2}\left\|\rho_{\Ext_{IP}\left(X_{I_{\ell}},Y_{I_{\ell}}\right)X^{\leq \ell-1}Y^{\leq \ell-1}\qstate}-\rho_{U_{q_\ell}}\otimes \rho_{X^{\leq \ell-1}Y^{\leq \ell-1}\qstate}\right\|
\leq  \sqrt{3}\cdot 2^{-\frac{1}{4}-\left(\frac{\delta}{4}-\frac{1}{8}\right)q_\ell n+2q_\ell}.
\end{align*}
Thus,
\begin{equation}\label{eq:proof2}
	\begin{aligned} 
\frac{1}{2}\big\|\rho_{U_{\tilde{q}_\ell}}\otimes\rho_{\Ext_{IP}\left(X_{I_{\ell}},Y_{I_{\ell}}\right)X^{\leq \ell-1}Y^{\leq \ell-1}\qstate}-\rho_{U_{\tilde{q}_\ell}}\otimes\rho_{U_{q_\ell}}\otimes \rho_{X^{\leq \ell-1}Y^{\leq \ell-1}\qstate}\big\|
\leq  \sqrt{3}\cdot 2^{-\frac{1}{4}-\left(\frac{\delta}{4}-\frac{1}{8}\right)q_\ell n+2q_\ell}.
\end{aligned}
\end{equation}

Finally, combining inequalities \eqref{eq:proof1} and \eqref{eq:proof2} and applying the triangle inequality, we conclude that
\begin{equation}
\begin{aligned}
 \frac{1}{2}\big\|\rho_{Z^{>\ell}\Ext_{IP}\left(X_{I_{\ell}},Y_{I_{\ell}}\right)X^{\leq \ell-1}Y^{\leq \ell-1}\qstate} -\rho_{U_{\tilde{q}_\ell}}\otimes\rho_{U_{q_\ell}}\otimes \rho_{X^{\leq \ell-1}Y^{\leq \ell-1}\qstate}\big\|
\leq \sqrt{3}\cdot \sum_{j=\ell}^k 2^{-\frac{1}{4}-\left(\frac{\delta}{4}-\frac{1}{8}\right)q_j n+2q_j}
\end{aligned}
\end{equation}
That is, 
\begin{equation*} 
\begin{aligned}
\frac{1}{2}\left\|\rho_{Z^{>\ell-1}X^{\leq \ell-1}Y^{\leq \ell-1}\qstate}-\rho_{U_{\tilde{q}_{\ell-1}}}\otimes \rho_{X^{\leq \ell-1}Y^{\leq \ell-1}\qstate}\right\|
\leq \sqrt{3}\cdot \sum_{j=\ell}^k 2^{-\frac{1}{4}-\left(\frac{\delta}{4}-\frac{1}{8}\right)q_j n+2q_j},
\end{aligned}
\end{equation*}
So  Eq.~\eqref{eq:induction} also holds for $\ell-1$. The proof is finished.
\end{proof}

As can be seen from the proof, to extract infinitely long randomness, the parameter $\Delta$ of $\Ext_{\rbs}^{\rmmneq}$ must be strictly positive, since the upper bound $\sqrt{3}\cdot \sum_{\ell=1}^k2^{-\frac{1}{4}-\left(\frac{\delta}{4}-\frac{1}{8}\right) q_\ell n+2q_\ell}$ on the error converges if and only if $\Delta>0$. Moreover, the infinitely long output string $Z^{(1)}\circ Z^{(2)}\circ\cdots$ can be arbitrarily close to the uniform distribution by choosing a sufficiently large constant $q_1$. 
In addition, $\Ext_{\rbs}^{\rmmneq}$ extracts $\frac{q_\ell}{2\delta q_\ell n}\approx \frac{2\delta-1}{48\delta}$ fraction of the min-entropy of the raw data $X_{I_{\ell}},Y_{I_{\ell}}$. Finally, we remark there is a tradeoff between the error and the computational complexity by setting $\Delta$. As $\Delta$ increases, the error becomes smaller, but the size of block increases, which leads to heavy computations.

We intuitively explain Theorem~\ref{thm:ext_rbs}, i.e., how our algorithms work for the forward block sources. When the input raw data $X,Y$ are forward block sources, the block $X_{I_{\ell}}$ and $Y_{I_\ell}$ are independent and each has large min-entropy conditioned on any assignment $(x^{<\ell},y^{<\ell})$ of previous blocks $(X^{<\ell},Y^{<\ell})$, which implies that the output $Z^{(\ell)}$ of this block is (approximately) uniform conditioned on any assignment $(x^{<\ell},y^{<\ell})$. Note that $z^{<\ell}$ can be fully determined by $(x^{<\ell},y^{<\ell})$. So $Z^{(\ell)}$ is (approximately) uniform conditioned on any assignment $z^{<\ell}$ of $Z^{<\ell}$. By induction on $\ell$, we can conclude that the whole output $Z^{(1)}\circ\cdots\circ Z^{(k)}\circ\cdots$ is (approximately) uniform. 

As a corollary of Theorem \ref{thm:ext_rbs}, we have the following result for $\Ext_{\rbs}^{\rmmeq}$.
\begin{theorem}\label{thm:ext_rbseq}
The output $Z^{(1)}\circ\cdots\circ Z^{(Nb/qn)}$ of $\Ext_{\rbs}^{\rmmeq}$ satisfies that
\begin{equation}
  \frac{1}{2}\|\rho_{Z^{(1)}\circ Z^{(2)}\circ\cdots\circ Z^{(Nb/qn)}\circ\qstate}-\rho_{U_{m}}\otimes \rho_{\qstate}\|\leq \epsilon,
\end{equation}
where $m:=\frac{Nb}{n}$. 
\end{theorem}

\section{Simulations of the real-time randomness generation rate}\label{sec:simulation}
In this section, we provide an analysis of the computational efficiency of our extractors. For simplicity of discussion, we focus on the first extractor  $\Ext_{\rbs}^{\rmmeq}$. 
\subsection{Asymtoptic analysis.} Suppose we apply $\Ext_{\rbs}^{\rmmeq}$ on two independent $(b,N,\delta)$-forward block sources. The extractor will set $n=O(1)$ and $q=O(\log N)$, and then compute the inner product function $\Ext_{IP}^{n,q}$ for $Nb/qn=O(Nb/\log N)$ times. The computation of $\Ext_{IP}^{n,q}$ involves $(n-1)$ additions `$+_q$' and $n$ multiplications `$\cdot_q$' over the finite field $\mathbb{F}_{2^q}$. The addition `$+_q$' is the bit-wise parity in $\{0,1\}^q$ and thus can be implemented by $q$ parity gates. The multiplication `$\cdot_q$' can be implemented by a $O(q\log q\log\log q)$-size Boolean circuit (see the last inequality in Section 2 and the third paragraph on Page 3 in \cite{DBLP:conf/aaecc/RifaB91}). So $\Ext_{IP}^{n,q}$ can be computed by a $O(nq\log q\log\log q)$-size Boolean circuit. Therefore, it requires $O(Nb/\log N)\cdot O(n q\log q\log\log q)=O(Nb\log\log N\log\log\log N)$ logic operations in total to execute $\Ext_{\rbs}^{\rmmeq}$.

We compare the computational complexity of $\Ext_{\rbs}^{\rmmeq}$ with the seeded extractor proposed in \cite{li2023improved}, which evenly partitions the input raw data into blocks each of size $O(\log N)$ and applies the Toeplitz-hashing extractor to each block separately with the same seed of $O(\log N)$ bits. Note that it needs $O(\log^2 N)$ logic operations to multiply a $O(\log N)\times O(\log N)$ dimensional Toeplitz matrix with a $O(\log N)$ dimensional vector. Thus, it requires  $O(Nb/\log N)\cdot O(\log^2 N)=O(Nb\log N)$ logic operations in total to execute the seeded extractor. So $\Ext_{\rbs}^{\rmmeq}$ is asymptotically faster than the seeded extractor.

\subsection{Simulations.} In the following, we make a simulation estimating the extraction speed of the first extractor $\Ext_{\rbs}^{\rmmeq}$ implemented in the Xilinx Kintex-7 XC7K480T Field Programmable Gate Array (FPGA), a common application in industry. 
We use the raw data from Ref.~\cite{gehring2021homodyne} as the random source, which is a forward block source with parameters $b=16$ and $\delta=10.74/16\approx 0.67$. We consider a raw data length of $N=2^{51}$ bits and a total security parameter $\epsilon=2^{-30}$, which means the final output data is $2^{-30}$-close to a uniform distribution. Correspondingly, the algorithmic parameters $n,q$ of $\Ext_{\rbs}^{\rmmeq}$ are specified to $n=\left\lceil 24/(2\delta-1)\right\rceil=71$ and $q:=b\left\lceil \log_2 \left(N/\epsilon n\right)/b\right\rceil=80$. By Theorem \ref{thm:ext_rbseq}, $\Ext_{\rbs}^{\rmmeq}$ uses no seed and will output about $0.45$ PB random bits.

Precisely, $\Ext_{\rbs}^{\rmmeq}$ simply divides the input raw data $X,Y$ into blocks each containing $71\times 80$ bits, and then apply $\Ext_{IP}^{71,80}$ on each block $(X_{I_\ell},Y_{I_\ell})$. The inner product function $\Ext_{IP}^{71,80}$ involves $71-1=70$ addition operations  and $71$ multiplication operations in the field $\mathbb{F}_{2^{80}}$. The addition operation in $\mathbb{F}_{2^{80}}$ is the bit-wise parity operation in $\{0,1\}^{80}$, and so can be computed in $80$ `$\oplus$' operations. Each multiplication operation in $\mathbb{F}_{2^{80}}$ can be computed by at most $4885$ `$\oplus$' or `$\wedge$'  operations \cite{DBLP:conf/aaecc/RifaB91}. Thus, each $\Ext_{IP}^{71,80}$ can be implemented  by at most $80\times (71-1)+ 4885\times 71= 352435$ `$\oplus$' or `$\wedge$'  operations.



The parameters of the Xilinx Kintex-7 XC7K480T FPGA are as follows. The clock rate is set to be $200$ MHz; the number of Look-Up-Tables (LUTs) is $3\times 10^5$; each LUT can perform $5$ basic logical operations (e.g., `$\oplus$' or `$\wedge$') simultaneously.  To make full use of the FPGA, we can perform the matrix multiplications of $\lfloor 3\times 10^5 \times 5/352435\rfloor=4$ blocks in parallel. 
Therefore, the extraction speed of $\Ext^{\rmmeq}_{\rbs}$ is $200\times 10^6\times 4\times 80 = 64$ Gbps, which is significantly improved by one order of magnitude compared to the state-of-the-art result with a speed of 18.8 Gbps \cite{bai202118}. As a result, 
our online extraction is adequate for the post-processing of the QRNG in \cite{gehring2021homodyne}.


\section{Conclusion}

In conclusion, we introduce a real-time two-source quantum-proof randomness extractor, a crucial advancement for applications such as quantum key distribution.
The key contribution lies in the design of online algorithms based on the concept of a forward block source. We demonstrate that a constant fraction of min-entropy can be extracted from two arbitrarily long independent forward block sources. Notably, this extractor accommodates a broad class of quantum random number generators. The inherent parallelizability of our extractor enhances its practicality for real-world applications.

Looking forward, there are several open questions. Firstly, our extractor currently necessitates a min-entropy rate $\delta > 1/2$. It remains an intriguing challenge to explore whether this restriction can be dropped or imitated without compromising the security and efficiency of the extraction process. Secondly, the design of an online extractor for more general sources, or even the most general min-entropy source, remains an open avenue for future research. These questions underscore the ongoing pursuit of robust quantum-proof randomness extraction techniques with broader applicability and enhanced flexibility. Our work lays a foundation for these inquiries, marking a significant step towards securing communication and cryptographic protocols in the quantum era.

\bibliographystyle{apsrev4-2}
\bibliography{ref}

\appendix
\section{Proof of Lemma \ref{lem:one-bit}}\label{sec:one-bit}
In this section, we prove Lemma \ref{lem:one-bit}. Let $X,Y$ be random variables on $(\{0,1\}^q)^n$ and $\gqstate=\gqstate_1\otimes \gqstate_2$ be the product quantum side information where $H_{\min}(X| \gqstate_1)\geq \delta qn$ and $H_{\min}(Y|\gqstate_2)\geq \delta qn$. Let $Z=\Ext_{IP}^{n,q}(X,Y)$, and fix $S\in\{0,1\}^q\setminus\{0^q\}$ arbitrarily. We will use ``$\times$" to denote the multiplication operation in the field $\mathbb{F}_{2^q}$. 

Note that for any $S\in\{0,1\}^q$, there exists a $a\in \mathbb{F}_{2^q}$ such that $S\cdot Z$ corresponds exactly to the first bit of $a\times Z$ viewed as a vector in $\{0,1\}^q$. Moreover this correspondence between $S$ and $a$ is a bijection in the sense that different $S\in\{0,1\}^q$ correspond to different $a\in\mathbb{F}_{2^q}$. The special case when $S=0^q$ corresponds to the case $a=0$. Thus, for any $S\in\{0,1\}^q\setminus\{0^q\}$, there exists some non-zero $a\in\mathbb{F}_{2^q}$ such that $S\cdot Z=[a\times Z]_1$. Here, $[a\times Z]_1$ is used to denote the first bit of $a\times Z$ viewed as a vector in $\{0,1\}^q$.

First, we show that the function $[a\times \Ext_{IP}^{n,q}(\cdot,\cdot)]_1$ is a hadamard function.
\begin{definition}[hadamard function]\label{def:hadamard}
A function $f:\{0,1\}^t\times\{0,1\}^t\rightarrow \{0,1\}$ is called a hadamard function if for any two distinct $x,x'\in\{0,1\}^t$, it has $\sum_{y\in\{0,1\}^t}(-1)^{f(x,y)+f(x',y)}=0$.
\end{definition}
\begin{proposition}\label{prop:ip}
For any non-zero $a\in\mathbb{F}_{2^q}$, $[a\times \Ext_{IP}^{n,q}(\cdot,\cdot)]_1$ is a hadamard function.
\end{proposition}
\begin{proof}
For simplicity of notations, we let $f_a(x,y):=[a\times \Ext_{IP}^{n,q}(\cdot,\cdot)]_1$. Given any two distinct $x=x_1\cdots x_n,x'=x_1'\cdots,x_{n}'\in (\{0,1\}^{q})^n$, w.l.o.g. we assume $x_1\neq x_1'$. Recalling that the addition operation $``+_q"$ in the field $\mathbb{F}_{2^q}$ corresponds to the bitwise-parity operation in $\{0,1\}^{q}$, we have 
\begin{align*}
\sum_{y\in(\{0,1\}^{q})^n} (-1)^{f_a(x,y)+f_a(x',y)}=&\sum_{y\in(\{0,1\}^{q})^n} (-1)^{f_a(x,y)-f_a(x',y)}=\sum_{y\in(\{0,1\}^{q})^n} (-1)^{f_a(x-x',y)}\\
=&\sum_{y\in(\{0,1\}^{q})^n} (-1)^{[a\times \Ext_{IP}^{n,q}(x-x',y)]_1}=\sum_{y\in(\{0,1\}^{q})^n} (-1)^{[a\times (x_1-x'_1)\times y_1]_1+\left[\sum_{i=2}^n a\times (x_i-x_i')\times y_i\right]_1}\\
=&\left(\sum_{y_1\in\{0,1\}^{q}} (-1)^{\left[a\times (x_1-x_1')\times y_1\right]_1}\right)\cdot \left(\sum_{y_2,\cdots,y_n\in\{0,1\}^q}(-1)^{\left[\sum_{i=2}^n a\times (x_i-x_i')\times y_i\right]_1}\right).
\end{align*}
Note that $a\times (x_1-x_1')\neq 0$, so 
\begin{align*}
\left|\left\{y_1\in\{0,1\}^q| \left[a\times (x_1-x_1')\times y_1\right]_1=1\right\}\right|
=\left|\left\{y_1\in\{0,1\}^q | \left[a\times (x_1-x_1')\times y_1\right]_1=0\right\}\right|=2^{q-1}.
\end{align*}
Thus we get the conclusion immediately by Definition \ref{def:hadamard}. 
\end{proof}

\begin{definition}[One-bit-output two-source extractor]
A function $f:\{0,1\}^{t}\times \{0,1\}^{t}\rightarrow \{0,1\}$ is a $(t,k,\epsilon)$ one-bit-output two-source extractor, if for every pair of independent sources $(X,Y)$ with min-entropy $\geq k$ each, we have
\[
(1-\epsilon)/2\leq \Pr[f(X,Y)=1]\leq (1+\epsilon)/2.
\]
\end{definition}
\begin{lemma}[\cite{CG88}]\label{thm:cg88}
Let $f:\{0,1\}^t\times\{0,1\}^t\rightarrow \{0,1\}$ be a hadamard function. Then $f$ is a $(t,k,2^{1-\frac{2k-t}{2}})$ one-bit-output two-source extractor for any $k\leq t$.

In other words, let $X,Y$ be independent random variables on $\{0,1\}^t$ with min-entropy $\geq k$ each. Then $(1-\epsilon)/2\leq \Pr[f(X,Y)=1]\leq (1+\epsilon)/2$ where $\epsilon=2^{1-\frac{2k-t}{2}}$.
\end{lemma}
So by the Lemma \ref{thm:cg88}, we know that the function $[a\times \Ext_{IP}^{n,q}(\cdot,\cdot)]_1$ is a $(qn,k,2^{1-\frac{2k-qn}{2}})$ one-bit-output two-source extractor for any $k\leq qn$.

Finally, Lemma \ref{lem:one-bit} follows from the Lemma \ref{lem:classical2quantum} immediately by letting $k=\left(\frac{1}{4}+\frac{\delta}{2}\right)qn+\frac{1}{2}$ and $\epsilon=2^{\frac{1}{2}-\left(\frac{\delta}{2}-\frac{1}{4}\right)qn}$.
\begin{lemma}[Corollary 27 of \cite{multi-source2} for $m=1$]\label{lem:classical2quantum}
Let $f$ be a $(t,k,\epsilon)$ one-bit-output two-source extractor. Then it is a $(t,k+\log \epsilon^{-1},\sqrt{3\epsilon/2})$-min-entropy quantum-proof two-source extractor. That is, for any random variables $X,Y$ on $\{0,1\}^t$ and any product quantum side information $\gqstate$ where 
$H_{\min}(X|\gqstate)\geq k+\log \epsilon^{-1}$ and $H_{\min}(Y| \gqstate)\geq k+\log \epsilon^{-1}$, we have $\frac{1}{2}\|\rho_{f(X,Y)\gqstate}-\rho_{U_1}\otimes\rho_{\gqstate}\|\leq \sqrt{3\epsilon/2}$.
\end{lemma}

\end{document}